\documentclass{scrartcl}
\usepackage[utf8x]{inputenc}
\usepackage[OT4]{fontenc}
\usepackage{amsmath,amsthm,latexsym,amssymb,amsfonts,color,graphicx}
\usepackage{bbm}
\usepackage{authblk}

\newcommand{\CP}{\mathbb{CP}}
\newcommand{\SP}{{\mathcal S}}
\newcommand{\SL}{\operatorname{SL}(2,\mathbb{C})}
\newcommand{\jj}{\jmath}
\newcommand{\Hodge}{\ast}
\newcommand{\RR}{I}

\newcommand{\Lv}{\mathbbm{L}}
\newcommand{\Rv}{\mathbbm{R}}
\newcommand{\zero}{{0'}}

\newcommand{\Li}{\vec{\mathbb L}}
\newcommand{\Ki}{\vec{\mathbb K}}
\newcommand{\geom}{\Delta}

\DeclareMathOperator{\Lie}{Lie}

\theoremstyle{theorem}

\newcommand{\R}{{\mathbb R}}
\newcommand{\C}{{\mathbb C}}

\newcommand{\M}{{\mathbb M}}

\newcommand{\LL}{{\mathcal L}}

\renewcommand{\bf}{\mathbbm}

\newcommand{\z}{{\bf z}}

\newcommand{\w}{{\bf w}}
\renewcommand{\u}{{\bf u}}
\renewcommand{\v}{{\bf v}}

\newcommand{\bmat}[1]{\left(
   \begin{array}{cc} #1
    \end{array}\right)}%

\newtheorem{thm}{Theorem}
\newtheorem{lm}{Lemma}
\newtheorem{df}{Definition}

\DeclareMathOperator{\tr}{tr}
\DeclareMathOperator{\Span}{span}

\begin{document}

\title{The hessian in spin foam models}

\author[1]{Wojciech Kami\'nski\thanks{wkaminsk@fuw.edu.pl}}
\author[2]{Hanno Sahlmann\thanks{hanno.sahlmann@gravity.fau.de}}
\affil[1]{\small Instytut Fizyki Teoretycznej, Wydzia{\l} Fizyki, Uniwersytet Warszawski, ul. Pasteura 5 PL-02093 Warszawa, Poland}
\affil[2]{\small Institute for Quantum Gravity, Department of Physics, Friedrich-Alexander Universit\"at Erlangen-N\"urnberg (FAU), 
Staudtstr. 7 D-91058 Erlangen, Germany}

\date{\today}

\maketitle

\abstract{We fill one of the remaining gaps in the asymptotic analysis of the vertex amplitudes of the Engle-Pereira-Rovelli-Livine (EPRL) spin foam models: We show that the hessian is nondegenerate for the stationary points that corresponds to geometric nondegenerate $4$ simplices. Our analysis covers the case when all faces are spacelike.}

\section{Introduction}

One of the central results of the research on spin foam models (defined in \cite{EPRL, FK} and extended in \cite{Conrady2010}) is the asymptotic analysis of the vertex amplitude accomplished in \cite{Conrady-Freidel, Frank3, FrankEPRL} for the euclidean case and in \cite{Frank2, Kaminski2018, Liu2019} for the lorentzian case). The graviton propagator \cite{Alesci2008, Bianchi2009, Bianchi2012}, the relation to Regge calculus and various semiclassical limits \cite{Han2011, Han2013a} are all based on this result. Let us mention that exactly the asymptotic analysis \cite{4dWilliams} of the vertex of the Barrett-Crane model \cite{Barrett-Crane} led to the discovery of nongeometric sectors \cite{Alesci2007} and in consequence to the invention of the EPRL model. However, it is important to keep in mind that the analysis of the vertex amplitude does not capture all properties of the model  -- as seen by so-called flatness problem \cite{Bonzom, Hellmann2013} that is not visible in the asymptotics of a single vertex.

The proof of the asymptotic formula for various spin foam models is not completely water-tight because of a few issues. First of all, the proof is based on stationary phase method and typically integration is done over noncompact domains. It is not clear if there are any contributions from infinity or from boundary of the domain of integration. In the Hnybida-Conrady extension \cite{Conrady2010} it is even not known if the amplitude is finite at all. Secondly, the contribution from a stationary point depends on whether the point is nondegenerate (i.e. the hessian at that point has no zero eigenvectors, after gauge fixing) or not. These issues were summarized in our previous paper \cite{Kaminski2018}.

The current paper is devoted to the problem of whether or not the hessian is nondegenerate for a given stationary point. The only analytic result in this direction that we know about for $4d$ models is the result \cite{Kaminski2013} for the Barrett-Crane model \cite{Barrett-Crane}. For the euclidean EPRL model it was checked for specific examples that the hessian is nondegenerate\footnote{Frank Hellmann, private communication.} so its determinant is nonzero for generic boundary data. However, the example of the Barrett-Crane model can serve as a warning, as in this case the hessian is degenerate for configurations where the map from lengths to areas of the $4$-simplex is not locally invertible. The lorentzian models are more complicated. The number of integration variables makes the determination of the determinant of the hessian an almost intractable task.

In this paper we will show that for the EPRL models  in both, euclidean and  lorentzian signature (we consider also Hnybida-Conrady extension), with spacelike faces, the hessian is nondegenerate for every stationary point that corresponds to a nondegenerate $4$-simplex (of either lorentzian or other signature). 

We will first consider the euclidean EPRL model with Barbero-Immirzi parameter $\gamma<1$, as it can be treated in a considerably simpler way. The crucial observation for our analysis of this case is the specific behavior of the hessian for actions satisfying a certain reality condition: If $e^{iS}$ denotes the integrand of the amplitude, then the imaginary part of the action is nonnegative, 
\begin{equation}
\label{eq:reality}
\Im S\geq 0.
\end{equation}
In order to extend our result to the case of lorentzian models we introduce a reduced action that is more closely related to the action of the euclidean model. The reduced action is defined in such a way that non-degeneracy of its hessian is equivalent to the non-degeneracy of that of the full action. We the reexpress the analysis of the euclidean amplitude in symplectic geometric terms. The geometric theory of such actions is based on positive lagrangeans that were introduced by \cite{Hormander-pos}. This makes it applicable to the lorentzian case as well. 

The main reference for our notation is \cite{Kaminski2018}. There are a few departures from that notation, for which we refer the reader to Appendix \ref{sec:appendix}.

\section{Euclidean EPRL model with $\gamma<1$.}

In the following, our terminology and, in particular what is real and what is imaginary is based on the convention that the integrand of the integral we are approximating is $e^{iS}$, and we will cal $S$ the action. We note that this is different from the convention of \cite{Frank2}.

For a symmetric (or hermitian) form $H$ we will use the notation
\begin{equation}
Hv=H(\cdot,v).
\end{equation}
We will say that the vector $v$ annihilates $H$ if
\begin{equation}
Hv=0.
\end{equation}
For a real symmetric form $I$ we write $I\geq 0$ if for any real vector $w$
\begin{equation}
I(w,w)\geq 0.
\end{equation}
This is equivalent to the condition that for any complex vector $v$
\begin{equation}
I(\bar{v},v)\geq 0.
\end{equation}

\begin{lm}\label{lm:pos}
Assume that the symmetric complex form $H$ can be decomposed as $H=R+iI$ where $R$ and $I$ are real and $I\geq 0$. Then the following conditions for a vector $v$ are equivalent:
\begin{enumerate}
\item $v$ annihilates $H$
\begin{equation}
Hv=0.
\end{equation}
\item The following is true for the real and imaginary part of the vector $v$ ($v=\Re v+i\Im v$):
\begin{equation}
R\Re v=R\Im v=I\Re v=I \Im v=0.
\end{equation}
\end{enumerate}
\end{lm}

\begin{proof}
Let us write $v=v_r+iv_a$ where $v_r$ and $v_a$ are real.

We have from the linearity of the forms
\begin{equation}
0=Hv=(Rv_r-Iv_a)+i(R v_a+Iv_r),
\end{equation}
thus $Rv_r=Iv_a$ and $Rv_a=-Iv_r$. Moreover from the symmetry of $R$
\begin{align}
I(v_a,v_a)=R(v_a,v_r)=R(v_r,v_a)=-I(v_r,v_r).
\end{align}
As $I\geq 0$ we see that $I(v_a,v_a)=0$ and $I(v_r,v_r)=0$, thus
\begin{equation}
Iv_r=Iv_a=0
\end{equation}
and also $Rv_r=Rv_a=0$.
\end{proof}

\begin{lm}\label{lm:sum}
Suppose that the symmetric real form $I=\sum_\alpha I_\alpha$ and $I_\alpha\geq 0$. Then
\begin{equation}
Iv=0\Longleftrightarrow \forall_\alpha I_\alpha v=0.
\end{equation}
\end{lm}

\begin{proof}
We have $I(\bar{v},v)=0$ thus $\sum_\alpha I_\alpha (\bar{v},v)=0$. All terms are positive, thus each of them needs to be zero, but due to positivity this implies that $I_\alpha v=0$.
\end{proof}

\subsection{Hessian in euclidean EPRL}

The manifold of integration is $\prod_{i=1}^4 Spin(4)$ and thus the vectors of the tangent space can be described by
\begin{equation}
v\colon \{1,\ldots,5\}\rightarrow \R^3\oplus \R^3,\quad v(5)=0.
\end{equation}
We will denote self-dual (anti-self-dual) part by $v^\pm$.

The tensor of second derivatives of the action (the hessian) is given by \cite{Frank3}\footnote{Published version.}
\begin{equation}
H(v,v')=H^+(v^+,{v'}^+)+H^-(v^-,{v'}^-).
\end{equation}
Let us consider the self-dual part (the antiself-dual is analogous). We can write $H^+$ as
\begin{equation}
H=R+i\left(\sum_{1\leq a<b\leq 5} I_{ab}\right),
\end{equation}
where $I_{ab}$ are given by
\begin{equation}
I_{ab}(v^+,v^+)=I'_{ab}(v^+(a)-v^+(b),v^+(a)-v^+(b))
\end{equation}
in terms of symmetric real forms $I'_{ab}\colon \R^3\times \R^3\rightarrow \R$
\begin{equation}
I'_{ab}(w,w)=\frac{j_{ab}^+}{2}\left(|w|^2-(w\cdot n_{ab}^+)^2\right).
\end{equation}
This form is $\frac{j_{ab}^+}{2}$ times the expectation value of the projector onto the space perpendicular to $n_{ab}^+$, so it is nonnegative ($\frac{j_{ab}^+}{2}\geq 0$), thus also $I_{ab}\geq 0$.

The real form $R$ is given by
\begin{equation}
R(v,v')=\sum_{a,b\in\{1,\ldots 5\}}\frac{j_{ab}^+}{2}n_{ab}\cdot v^+(a)\times {v'}^+(b),
\end{equation}
where we use the convention that $n_{ab}=-n_{ba}$.

\begin{lm}
If $\det H=0$ then there exist $a\not=b\in\{1,\ldots 4\}$ such that $n_{a5}$, $n_{ab}$, $n_{b5}$ are linearly dependent.
\end{lm}

\begin{proof}
If $\det H=0$ then there exists a nonzero vector $v'$ such that $Hv'=0$ thus Lemma \ref{lm:pos} assures that there exists a nonzero real vector $v$ that is annihilated by $\Im H$. It needs to be annihilated by every $I_{ab}$ due to Lemma \ref{lm:sum}. The conditions
\begin{equation}
I_{5a}v=0,\quad I_{5b}v=0
\end{equation}
give
\begin{equation}
v(a)=\lambda_bn_{a5},\quad v(b)=\lambda_an_{b5},
\end{equation}
where $\lambda_k\in \R$. The condition $I_{ab}v=0$ gives
\begin{equation}
v(a)-v(b)=\lambda_5 n_{ab},
\end{equation}
thus
\begin{equation}
\lambda_bn_{a5}-\lambda_an_{b5}-\lambda_5 n_{ab}=0.
\end{equation}
Either $v(a)=v(b)=0$ or $n_{ab},n_{a5},n_{b5}$ are linearly dependent. 

As this is true for all $a,b$ we have either $v=0$ (contradiction) or there exist $a,b$ fulfilling the statement of the lemma.
\end{proof}

\begin{thm}
The hessian for the euclidean EPRL model with $\gamma<1$ is nondegenerate for any stationary point that corresponds to a nondegenerate $4$-simplex.
\end{thm}

\begin{proof}
If $n_{ab},n_{a5}, n_{b5}$ are linearly dependent then the matrix $\tilde{G}_{ab5}$ defined in equation (301) from \cite{Kaminski2018} is degenerate, and lemma 28 from \cite{Kaminski2018} (in its version for euclidean signature) tells us that there exists at most one stationary point (a single vector geometry or a degenerate $4$-simplex). 
\end{proof}

For the case of euclidean EPRL just considered the integration is over the compact manifold, thus the nondegeneracy of the hessian was the only missing part of the asymptotic analysis. We will not consider euclidean case with $\gamma>1$ because it can be treated in an analogous way to the lorentzian case. We will  now describe the lorentzian case in detail.

\section{Extension to the lorentzian EPRL amplitude}\label{sec:extension}

In the case of the lorentzian EPRL amplitude, integration is over many more variables and the hessian is more complicated. The action is a sum\footnote{We will use the notation from \cite{Kaminski2018}. A summary of notation and conventions is also in Appendix  \ref{sec:appendix}.}
\begin{equation}
\tilde{S}(\{g_i\},\{\z_{ij}\})=\sum_{1\leq i<j\leq 5} \tilde{S}_{ij}(g_i,g_j,\z_{ij},\z_{ji}),
\end{equation}
where
\begin{equation}
\tilde{S}_{ij}(g_i,g_j,\z_{ij},\z_{ji})=S_{ij}^{n_{ij}}(g_i^{-1}\z_{ij}) +S^\beta_{ij}(\z_{ij},\z_{ji})+S_{ji}^{n_{ji}}(g_j^{-1}\z_{ji}).
\end{equation}
Actions as well as measure factors are at least locally analytic. If we denote by $[\z_{ij}]$ elements of $\CP$ (i.e., equivalence classes of spinors) then the stationary points are discrete and we are interested in one of them
\begin{equation}
g_i^0,\ [\z_{ij}^0],
\end{equation}
where $g_5^0=1$. We will denote bivectors (see Section \ref{symp:G} and Appendix \ref{sec:appendix} for notation\footnote{It differs slightly from \cite{Kaminski2018} due to other normalization constants in the scalar product and some sign factors.})
\begin{equation}
B_{ij}^0=\delta_{\z_{ij}} S^\beta_{ij},
\end{equation}
and we will write $B_{ij}^\zero=g_i^{-1}B_{ij}^0$ for a bivector in the node frame. We will call it the fundamental stationary point.

\subsection{Reduced action}\label{sec:reduced}

The variables $\{g_i\}$ appear in many places, but for fixed $ij$ the variables $\z_{ij}$ and $\z_{ji}$ $\in$ $\CP$ are only found in the action $\tilde{S}_{ij}$. Let us denote the form of second derivatives with respect to the $\CP$ variables by $H_{\z\z}$. It is block diagonal, with blocks corresponding to $\{\z_{ij},\z_{ji}\}$. We will show later that this form is nondegenerate (in the neighbourhood of the fundamental stationary point).

Let us (locally) analytically extend the action in the $\z$ variables to the complexification $\CP^\C$, \footnote{We regard $\SL$ and $\CP$ as real manifolds. Complexification of a space $M$ that is already a complex manifold gives a space $M\times \overline{M}$}
\begin{equation}
\tilde{S}(\{g_i\},\{\z_{ij}^\C\}).
\end{equation}
Let us notice that $\tilde{S}$ depends only on the $\CP$ variables $[\z_{ij}^\C]$ (equivalence classes of spinors).
As the hessian $H_{\z\z}$ is nondegenerate at the fundamental stationary point we can (in the neighbourhood of $g_i^0$) find a unique (in the neighbourhood of $[\z_{ij}]^0$) solution
\begin{equation}
\label{eq:statz}
[\z_{ij}^\C]\colon \forall_{ij} \frac{\partial \tilde{S}}{\partial [\z_{ij}^\C]}=0.
\end{equation}
Here $\frac{\partial \tilde{S}}{\partial [\z_{ij}^\C]}$ is a holomorphic derivative as the antiholomorphic one gives $\frac{\partial \tilde{S}}{\partial \overline{[\z_{ij}^\C]}}=0$ everywhere. Let us notice that due to the form of the action the solution has a specific dependence on $\{g_i\}$
\begin{equation}
[\z_{ij}^\C](g_i,g_j).
\end{equation}
Let us introduce a reduced action
\begin{equation}
S^{red}(\{g_i\})=\sum_{i<j} S_{ij}^{red}(g_i,g_j),\ \ S_{ij}^{red}(g_i,g_j)=\tilde{S}_{ij}(g_i,g_j,[\z_{ij}^\C](g_i,g_j),[\z_{ji}^\C](g_j,g_i)).
\end{equation}
The point $g_i^0$ is a stationary point of this action and the hessian at this point is $H^{red}=\sum H_{ij}^{red}$. Let us notice that $S^{red}_{ij}$ depends only on the group element  $g_{ij}=g_j^{-1}g_i$ 
\begin{equation}
S^{red}_{ij}(g_i,g_j)=S'_{ij}(g_{ij}).
\end{equation}
We have the projection map on the complexified tangent space 
\begin{equation}
\Pi\colon T^\C_{\{g_i^0,[\z_{ij}^0]\}}\left(\prod_i \SL\times \prod_{i\not=j}\CP\right)\rightarrow T^\C_{\{g_i^0\}}\left(\prod_i \SL\right).
\end{equation}
We can also introduce a cross section
\begin{align}
&\Xi\colon  T^\C_{\{g_i^0\}}(\prod_i \SL)\rightarrow T^\C_{\{g_i^0,[\z_{ij}^0]\}}(\prod_i \SL\times \prod_{i\not=j}\CP),\\
&\Xi(V)=V+\sum_{ij} V\z_{ij}^\C(g_i,g_j)\frac{\partial}{\partial [\z_{ij}^\C]}+V\overline{\z_{ij}^\C}(g_i,g_j)\frac{\partial}{\partial [\overline{\z_{ij}^\C}]}.
\end{align}
We also use these maps restricted to fixed $ij$ sectors ($\Pi_{ij}$ and $\Xi_{ij}$).

\begin{lm}
The following holds:
\begin{equation}
H^{red}_{ij}(\Pi_{ij}(W_{ij}),V_{ij})=H_{ij}(W_{ij},\Xi_{ij}(V_{ij})).
\end{equation}
Also
\begin{equation}
H^{red}(\Pi(W),V)=H(W,\Xi(V)).
\end{equation}
\end{lm}

\begin{proof}
Due to the condition \eqref{eq:statz} on $\z^\C$ we have for $W_{ij}\in (T\SL)^2\times (T\CP)^2$
\begin{align}
&W_{ij}\left(\Im S_{ij}(g_i,g_j,\z_{ij}^\C(g_i,g_j), \z_{ji}^\C(g_i,g_j))\right)=\\
&=\left(\Pi_{ij}(W_{ij})\Im S_{ij}\right)(g_i,g_j,\z_{ij}^\C(g_i,g_j), \z_{ji}^\C(g_i,g_j)).
\end{align}
Let us notice that for $V_{ij}\in T\SL^2$
\begin{align}
&V_{ij}\left((\Pi_{ij}(W_{ij})\Im S_{ij})(g_i,g_j,\z_{ij}^\C(g_i,g_j), \z_{ji}^\C(g_i,g_j))\right)=\\
&=V_{ij}\left(W_{ij}\Im S_{ij}\right)(g_i,g_j,\z_{ij}^\C(g_i,g_j), \z_{ji}^\C(g_i,g_j))=\\
&=\left(\Xi_{ij}(V_{ij})W_{ij}\Im S_{ij}\right)(g_i,g_j,\z_{ij}^\C(g_i,g_j), \z_{ji}^\C(g_i,g_j)).
\end{align}
Thus $H^{red}_{ij}(V_{ij},\Pi_{ij}(W_{ij}))=H_{ij}(\Xi_{ij}(V_{ij}),W_{ij})$. Summing over $ij$ we get also the second equality.
\end{proof}

\begin{lm}
The hessian is degenerate if and only if the reduced hessian is.
\end{lm}

\begin{proof}
Let suppose that $HV=0$, then for any $W$
\begin{equation}
0=H(\Xi(W),V)=H^{red}(W,\Pi(V)),
\end{equation}
thus $H^{red}\Pi(V)=0$. The other way around, if $H^{red}W=0$ then for any $V$
\begin{equation}
H(\Xi(W),V)=H^{red}(W,\Pi(V))=0,
\end{equation}
thus $H\Xi(W)=0$.
\end{proof}

\begin{df}
An \emph{extremal point} of the action $S$ is a point on the real manifold where $\partial \Im S=0$ and the tensor of second derivatives of $\Im S$ is nonnegative definite.
\end{df}

If the action $S$ satisfies the reality condition \eqref{eq:reality} ($\Im S\geq 0$) then points on the real manifold where $\Im S=0$ are extremal. The fundamental stationary point $\{g_i^0,[\z_{ij}]^0\}$ is extremal for the actions $\tilde{S}_{ij}$. The following is a consequence:

\begin{lm}
The hermitean form
\begin{equation}
\Im H_{ij}^{red}(\overline{V_{ij}},V_{ij})
\end{equation}
is nonnegative definite.
\end{lm}

\begin{proof}
The maps $\Xi_{ij}$ and $\Pi_{ij}$ are compatible with complex conjugation thus
\begin{align}
&\Im H_{ij}^{red}(\overline{V_{ij}},V_{ij})=\Im H_{ij}^{red}(\overline{V_{ij}},\Pi_{ij}\Xi_{ij}(V_{ij}))=\\
&=\Im H_{ij}(\Xi_{ij}(\overline{V_{ij}}),\Xi_{ij}(V_{ij}))=
\Im H_{ij}(\overline{\Xi_{ij}(V_{ij}}),\Xi_{ij}(V_{ij}))\geq 0,
\end{align}
because the imaginary part of the hessian $H_{ij}$ is nonnegative definite.
\end{proof}

Let us summarize:

\begin{lm}
The point $\{g_i^0\}$ is an extremal point of $S_{ij}^{red}$.
\end{lm}

\section{Symplectic geometry}

We will adapt the theory of positive lagrangeans introduced in \cite{Hormander-pos}.
Let $\Omega$ be the symplectic form on $T^*M$. It is the inverse to the Poisson bracket
\begin{equation}
\Omega(v,\{D,\cdot\})=v(D),\quad v\in T(T^*M),\ D\in C^\infty(T^*M).
\end{equation}
Let us consider an analytic function $S\colon M\rightarrow \C$ (maybe defined only on an open set $U$).
The manifold
\begin{equation}
\LL_S=\{(x,p)\colon \theta=dS(x)\}\subset T^{*\C}M
\end{equation}
is lagrangean, that is it extends analytically to an analytic lagrangean submanifold of $T^*M^\C$ in some neighbourhood of the real $T^*M$. Here we denoted by $\theta$ the tautological form $\theta=p_\mu dx^\mu$.

Over real points of $M$ the complex conjugation of the tangent space of the lagrangean $\overline{T^\C\LL_S}$ is in itself the tangent space of the holomorphic lagrangean
\begin{equation}
\LL_{\bar{S}}=\{(x,p)\colon \theta=d\bar{S}(x)\}\subset T^{*\C}M.
\end{equation}
The tangent space of $\LL_S$ can be identified by projection $\pi\colon T^*M\rightarrow M$ with the tangent space of $M$. We will denote this map by $\Pi_S\colon T^\C\LL\rightarrow T^\C M$.

Now we will state and prove some important facts about extremal points:

\begin{lm}\label{lm-real}
The following holds for an extremal point $x_0$ of the action $S$
\begin{enumerate}
\item $p_0=dS(x_0)$ is real.
\item The hermitian form on $T^\C_{(x_0,p_0)}\LL$
\begin{equation}\label{eq:R}
\RR(v,v')=-\frac{i}{2}\Omega_{(x_0,p_0)}(\bar{v},v'),\quad v,v'\in T_{(x_0,p_0)}^\C\LL
\end{equation}
is nonnegative definite and
\begin{equation}
\RR(v,v')=\partial^2\Im S(\Pi_S\bar{v},\Pi_Sv').
\end{equation}
\item Let $w\in T_{x_0}M^\C$ and we denote $v=\Pi_S^{-1}w$ then
\begin{equation}
\RR v=0 \; (\text{that is: } (\partial^2\Im S)w=0\;)
\end{equation}
is equivalent to 
\begin{equation}
v\in T_{(x_0,p_0)}^\C\LL_S\cap \overline{T_{(x_0,p_0)}^\C\LL_S}.
\end{equation}
\end{enumerate}
\end{lm}

\begin{proof}
At an extremal point $p=d\Re S$ because derivatives of imaginary parts vanish.
Let us use local coordinates $p_\mu,x^\mu$ on $T^*M$ then
\begin{equation}
\left\{p_\mu-\frac{\partial \bar{S}}{\partial x^\mu}, p_\nu-\frac{\partial S}{\partial x^\nu}\right\}=2i\frac{\partial^2\Im S}{\partial x^\mu\partial x^\nu}.
\end{equation}
Every vector tangent to $\LL$ can be written as
\begin{equation}
V=f^\mu\{p_\mu-\partial_\mu S,\cdot\},
\end{equation}
where $f^\mu$ are some complex constants. Thus at the point $(x_0,p_0)$
\begin{equation}
-\frac{i}{2}\Omega(\bar{V},V)=-\frac{i}{2}\overline{f^\mu}f^\nu\{p_\mu-\partial_\mu \bar{S},p_\nu-\partial_\nu S\}=\overline{f^\mu}f^\nu \partial_\mu\partial_\nu\Im S.
\end{equation}
From tensoriality of the second derivative at a point where $\partial \Im S=0$ we get
\begin{equation}
-\frac{i}{2}\Omega(\bar{V},V)=\partial^2\Im S(\overline{\Pi_SV},\Pi_SV),
\end{equation}
thus it is nonnegative definite. Let $V\in T^\C\LL_S$ be such that $\RR V=0$ then
\begin{equation}
\Omega(V,W)=0
\end{equation}
for all $W\in \overline{T^\C\LL_S}=T^\C\LL_{\overline{S}}$ thus $V\in T^\C\LL_{\overline{S}}$.
\end{proof}

We need some definition.

\begin{df}
We will say that the lagrangean $\LL$ at the real point $(x_0,p_0)$ is positive if
\begin{equation}
\RR_{(x_0,p_0)}(v,v')=-\frac{i}{2}\Omega(\bar{v},v')
\end{equation}
is nonnegative definite. We will say that it is strictly positive if additionally $\RR_{(x_0,p_0)}$ is nondegenerate (has no zero vectors).
\end{df}

A lagrangean is strictly positive if and only if
\begin{equation}
T_{(x_0,p_0)}^\C\LL\cap\overline{T_{(x_0,p_0)}^\C\LL}=\{0\},
\end{equation}
that is, the only real vector in $T^\C_{(x_0,p_0)}\LL$ is the trivial vector.

Let $f$ be an analytic function on $T^*M$ (it extends locally to $T^*M^\C$) that vanishes on $\LL_S$. The complex vector field
\begin{equation}
\{f,\cdot\}
\end{equation}
is tangent to $\LL_S$. If at the real point the lagrangean is positive then
\begin{equation}
-i\{\bar{f},f\}=-i\Omega(\{\bar{f},\cdot\},\{f,\cdot\})\geq 0.
\end{equation}

\subsection{Symplectic theory of $T^*\SL$}\label{symp:G}

The left invariant vector field $\Lv(L)$ of the Lie algebra element $L$ corresponds to the first order jets of $g\rightarrow ge^{tL}$. The right invariant vector field $\Rv(L)$ of the same Lie algebra element will be $g\rightarrow e^{-tL}g$ (the sign is necessary for proper commutation relations).

With every point of the cotangent bundle $T^*\SL$ we can associate a left and a right coalgebra element $p^L$ and $p^R$ given by the formula
\begin{equation}
\forall L\in so(1,3)\quad p^L(L)=\theta(\Lv(L)),\quad p^R(L)=\theta(\Rv(L)).
\end{equation}
Let us notice that at the base point $g$
\begin{equation}
p^R=-g^{-1}\cdot p^L,
\end{equation}
where $g$ acts on the coalgebra by the co-adjoint action (if we identify the coalgebra with bivectors using the scalar product then the coadjoint action is the same as the adjoint action, see appendix \ref{sec:appendix}).
For any Lie algebra element $L$,
\begin{equation}
p^L(L),\ p^R(L)
\end{equation}
are functions on $T^*\SL$. We have
\begin{align}\label{eq:comu}
&\{p^L(L),p^L(L')\}=-p^L([L,L']),\ \{p^R(L),p^R(L')\}=-p^R([L,L']),\\
& \{p^L(L),p^R(L')\}=0,
\end{align}
and also
\begin{equation}
\{f(g),p^L(L)\}=\Lv(L)f,\ \{f(g),p^R(L)\}=\Rv(L)f.
\end{equation}
Let us denote by $\delta^L S$ ($\delta^R S$) the covectors identified by with coalgebra as follows
\begin{equation}
\delta^LS(L)=\Lv(L)S,\quad \delta^RS(L)=\Rv(L)S. 
\end{equation}
We will use $\delta$ for the left version. We can use the standard scalar product $(\cdot,\cdot)$ on bivectors to make the further identification of $\delta S$ with a bivector.\footnote{This introduces additional factor of $2$ in comparison to \cite{Kaminski2018} (see section \ref{sec:appendix}).}

For any function $S$ on the group we can now define a lagrangean submanifold
\begin{equation}
\LL_S=\{\theta=d S\}=\{p^L=\delta S\}=\{p^R=\delta^R S\}.
\end{equation}

\subsection{Symplectic theory of a coadjoint orbit}\label{sec:coadjoint}

Let us recall that we can identify the space of bivectors (Lie algebra $so(1,3)=\Lambda^2\R^4$) with the coalgebra using the natural scalar product $(\cdot ,\cdot)$ on bivectors. Let us consider a coadjoint orbit
\begin{equation}\label{eq:X}
X_{n,\rho}=\left\{B\in \Lambda^2\R^4\colon (B,B)=\frac{1}{4}(n^2-\rho^2),\ (B,\Hodge B)=-\frac{1}{2}\rho n\right\},
\end{equation}
where $C_1=(B,B)$ and $C_2=(B,\Hodge B)$ are two Casimirs (invariants). The Lorentz group acts transitively on $X_{n,\rho}$. 
We have a natural Poisson bracket given, for a linear function $H(L)(B)=(B,L)$, by
\begin{equation}
\{H(L),H(L')\}=-H([L,L']).
\end{equation}
This turns the coadjoint orbits into symplectic manifolds.
Let us introduce an isomorphism from $so(1,3)$ to $sl(2,\C)$ (traceless matrices) by (see appendix \ref{sec:appendix})
\begin{equation}
B\rightarrow \M(B), \quad \M(v\wedge v')=\frac{1}{4}\left(\eta^-(v)\eta^+(v')-\eta^-(v')\eta^+(v)\right).
\end{equation}
We have identity 
\begin{equation}
-2\tr \M(B)^2=(B,B)-i(B,\Hodge B),
\end{equation}
thus for $B\in X_{n,\rho}$ 
\begin{equation}
\frac{1}{2}\tr \M(B)^2=\left(\frac{1}{4}(\rho-in)\right)^2.
\end{equation}
For the matrix $\M(B)$ there exist two spinors $\z^\pm_B$ (unique up to a constant each) such that
\begin{equation}
\M(B)\z^\pm_B=\pm \frac{1}{4}(\rho-in)\z_B^\pm.
\end{equation}
We can thus define a projection
\begin{equation}
\pi\colon X_{n,\rho}\rightarrow \CP,\ \pi(B)=[\z^+_B].
\end{equation}

\begin{df}
A function
\begin{equation}
S\colon U\subset \C^2\setminus\{0\}\rightarrow \C \text{ modulo } 2\pi
\end{equation}
is of type $(n,\rho)$ if
\begin{equation}\label{eq:S}
S(r e^{i\phi}\z)=S(\z)+\rho\ln r+n\phi.
\end{equation}
\end{df}

Usually we cannot define such actions globally. 
Let us introduce the notation\footnote{Our notation differs from \cite{Kaminski2018}.} (where the action is on Weyl spinors $\SP^+$, see appendix \ref{sec:appendix})
\begin{equation}
\delta_\z f(\z)=\delta_g^L f(g^{-1}\z)|_{g=1}.
\end{equation}
Let us notice that $\delta_\z S$ is a well defined function on $\CP$ if $S$ satisfies \eqref{eq:S}.

For a given $\z\in \C^2\setminus\{0\}$ we can consider a group $H_{[\z]}\subset \SL$ that preserves $[\z]\in \CP$. The Lie algebra of this group is given by
\begin{equation}
\Lie H_{[\z]}=\{B\colon [\z,\M(B)\z]=0\},
\end{equation}
where $[\u,\v]=\u^T\omega \v$ (see appendix \ref{sec:appendix}).
The subgroup that preserves $\z$ is $H^0_{[\z]}$
\begin{equation}
\Lie H^0_{[\z]}=\{B\colon \M(B)\z=0\}=\{B\colon \exists\lambda\in \C,\ \M(B)=\lambda\z\z^T\omega\}.
\end{equation}
Let us notice that $B\in {\Lie H_{[\z]}^0}^\perp$ is equivalent to
\begin{equation}
\forall\lambda\in\C\colon 0=\Re \tr\M(B)\lambda\z\z^T\omega=\Re \lambda[\z,\M(B)\z],
\end{equation}
thus to $B\in \Lie H_{[\z]}$. Moreover, the scalar product is null on $\Lie H^0_{[\z]}$. 

The functions of type $(0,0)$ are special as they can be pushed forward to $\CP$. For such $f$ we will denote $[f]_\CP$ such push forward, thus
\begin{equation}
f=[f]_\CP\pi.
\end{equation}
We have the action of $\SL$ on $\SP^+$ (and on $\CP$ respectively) generated by vector fields $\Lv_{\SP^+}(L)$ ($\Lv_\CP(L)$ respectively) for $L\in so(1,3)$. Vector fields $\Lv_{S^+}(L)$ (respectively $\Lv_\CP(L)$) correspond to the jet of the curves
\begin{equation}
t\rightarrow e^{-tL}\z,\quad (t\rightarrow [e^{-tL}\z]\text{ respectively}).
\end{equation}
Let us consider a map 
\begin{equation}
T^*_{[\z]}\CP\ni p\rightarrow \phi_{[\z]}(p)\in so(1,3),\quad \forall_L (\phi_{[\z]}(p),L)=p(\Lv_\CP(L))\text{ at point }[\z].
\end{equation}
Let us notice that for $f\in C^\infty(\CP)$ we have 
\begin{equation}
\delta_\z (f\pi)(\z)=\phi_{[\z]}(df).
\end{equation}

\begin{lm}
The map $\phi_{[\z]}$ is a bijection from $T^*_{[\z]}\CP$ to $\Lie H_{[\z]}^0$.
\end{lm}

\begin{proof}
Let us notice that $\Lv_\CP(L)([\z])=0$ if and only if $L\in \Lie H_{[\z]}$ thus $\phi_{[\z]}\in {\Lie H_{[\z]}}^\perp=\Lie H_{[\z]}^0$. As $so(1,3)_\CP$ span the whole tangent space at $[\z]$ we have also injectivity.
\end{proof}

\begin{lm}
For $S$ of type $(n,\rho)$ we have
\begin{equation}
\delta_\z S(\z)\in \Lie H_{[\z]}\cap X_{n,\rho}
\end{equation}
and $\pi(\delta_\z S(\z))=[\z]$.
\end{lm}

\begin{proof}
For any $L\in \Lie H_{[\z]}^0$ we have $\delta_\z S(L)=0$ thus $\delta_\z S\in {\Lie H_{[\z]}^0}^\perp=\Lie H_{[\z]}$. Let us consider a traceless matrix (for some spinor $\u$)
\begin{equation}
{\mathbb N}=\frac{1}{2}(\u\z^T+\z\u^T)\omega,
\end{equation}
then ${\mathbb N}\z=\frac{1}{2}[u,\z]\z$. Let furthermore $\M(B)={\mathbb N}$, then we have
\begin{equation}
(\delta_\z S,B)=-\frac{1}{2}(\Re [u,\z])\rho-\frac{1}{2}(\Im [u,\z])n=-\Re \frac{1}{2}(\rho-in)[\u,\z].
\end{equation}
As $\M(\delta_\z S)\in \Lie H_{[\z]}$ we can write
\begin{equation}
\M(\delta_\z S)=\frac{1}{2}(\v\z^T+\z\v^T)\omega
\end{equation}
and
\begin{equation}
(\delta_\z S,B)=-2\tr \M(\delta_\z S){\mathbb N}=-[\v,\z][\u,\z],
\end{equation}
so $[\v,\z]=\frac{1}{2}(\rho-in)$ and as
\begin{equation}
\M(\delta_\z S)\z=\frac{1}{2}(\rho-in)\z,\quad2\tr \M(\delta_\z S)^2=[\v,\z]^2=\frac{1}{4}(\rho-in)^2,
\end{equation}
thus $\delta_\z S\in X_{n,\rho}$ and $\pi(\delta S)=[\z]$.
\end{proof}

\begin{lm}
For any real function $S$ of type $(n,\rho)$ the map
\begin{equation}
([\z],p)\rightarrow \delta_\z S+\phi_{[\z]}(p)
\end{equation}
is a symplectic diffeomorphism from $T^*\CP$ to $X_{n,\rho}$. This map is compatible with the projection onto $\CP$. 
\end{lm}

\begin{proof}
Let us choose $f\in C^\infty(\CP)$ such that $df([\z])=p$ then
\begin{equation}
\delta_\z S+\phi_{[\z]}(p)=\delta_\z (S+f\pi)\in X_{n,\rho},
\end{equation}
as $S+f\pi$ is of type $(n,\rho)$. Moreover
\begin{equation}
\M(\delta_\z(S+f\pi))(\z)=\frac{1}{4}(\rho-in)\z+0\z,
\end{equation}
thus $\pi(\delta_\z (S+f\pi))=[\z]$, so it is compatible with the projection on $\CP$.
 
If $B\in X_{n,\rho}$ and $\pi(B)=[\z]$ then
 \begin{equation}
B-\delta_\z S(\z)\in \Lie H_{[\z]}^\perp=\Lie H_{[\z]}^0.
\end{equation}
However $\phi_{[\z]}$ is a bijection onto $\Lie H_{[\z]}^0$.
 
In order to check that it is a symplectomorphism we will show that Poisson brackets between generators of $so(1,3)$ are right. For $L\in so(1,3)$ let us consider the pull back of the Hamiltonian $H(L)$ to $T^*\CP$. It is
\begin{equation}
\delta_\z S(\z)(L)+\theta(L_\CP)(\z)=\Lv_{\SP^+}(S)(\z)+\theta(\Lv_\CP(L))(\z).
\end{equation}
Let us notice that $\Lv_{\SP^+}(S)$ descents to a function $[\Lv_{\SP^+}(S)]_\CP$ on $\CP$. We have thus for a given bivector a function on $T^*\CP$
\begin{equation}
H_\CP(L)=[\Lv_{\SP^+}(L)(S)]_\CP+\theta(\Lv_\CP(L)).
\end{equation}
Let us notice that
\begin{align}
&\{[\Lv_{\SP^+}(L)(S)]_\CP,[\Lv_{\SP^+}(L')(S)]_\CP\}=0,\\
&\{\theta(\Lv_\CP(L)),\theta(\Lv_\CP(L'))\}=-\theta(\Lv_\CP([L,L'])).
\end{align}
Moreover
\begin{equation}
\{[\Lv_{\SP^+}(L)(S)]_\CP,\theta(\Lv_\CP(L'))\}=[\Lv_{\SP^+}(L')\Lv_{\SP^+}(L)(S)]_\CP,
\end{equation}
thus
\begin{equation}
\begin{split}
\{[\Lv_{\SP^+}(L)(S)]_\CP,\theta(\Lv_\CP(L'))\}-&\{[\Lv_{\SP^+}(L')(S)]_\CP,\theta(\Lv_\CP(L))\}=\\
&\qquad\qquad=[\Lv_{\SP^+}([L',L])(S)]_\CP.
\end{split}
\end{equation}
Therefore finally
\begin{equation}
\{H(L),H(L')\}=-H([L,L']).
\end{equation}
Because the Hamiltonian vector fields of functions span in every point the whole tangent space, $\Omega$ is the same as the canonical symplectic form on the cotangent bundle.
\end{proof}

Let us now consider a complex action (locally defined) $S$ of type $(n,\rho)$. Let us notice that $\Im S$ is a function on $\CP$. In particular $\partial^2[\Im S]_\CP$ is a tensor on $\CP$.

\begin{lm}
The space
\begin{equation}
\LL_S'=\{\delta_\z S\colon [\z]\in U\}
\end{equation}
is a complex lagrangean manifold in $X_{n,\rho}$ and on the real point $B\in X_{n,\rho}$
\begin{equation}
\RR_{B}(v,v')=\partial^2[\Im S([\z])]_\CP(\overline{\pi(v)},\pi(v')),
\end{equation}
where $v\in T\LL_S'$ and $[\z]=\pi(B)$. 
\end{lm}

\noindent{\textit Remark:} We regard $\CP$ as a real manifold, thus $\pi(v)\in T^\C\CP$ and the conjugation is with respect to this additional complex structure. It can be translated into inner complex conjugation.

\begin{proof}
Let $S_{aux}$ be an auxiliary real action of type $(n,\rho)$. The difference $f=S-S_{aux}$ is a well defined function on $\CP$. Moreover using the local identification of $X$ with $\CP$ we have
\begin{equation}
\LL_S'=\{\theta=d[f]_\CP\}.
\end{equation}
Indeed this is equivalent to
\begin{equation}
B=\delta S_{aux}+\phi_{[\z]}(p)=\delta S_{aux}+\phi_{[\z]}(d[f]_\CP)=
\delta S_{aux}+\delta f=\delta S.
\end{equation}
We know that
\begin{equation}
\partial^2[\Im S]_\CP(\overline{\pi(v)},\pi(v'))=\partial^2[\Im (S-S_{aux})]_\CP(\overline{\pi(v)},\pi(v')=\RR(v,v'),
\end{equation}
thus the result.
\end{proof}

\subsection{Casimir reduction}

Let us consider a symplectic reduction of $T^*\SL$ with respect to Casimirs. 
For $\SL$ the moment map is nondegenerate except for bivectors equal to zero.

\begin{lm}
Two points $(g,p)$ and $(g',p')$ are connected by a flow of Casimirs in $\SL$ if and only if there exists $\lambda,\lambda'\in \R$ such that
\begin{equation}
g=g'e^{\lambda p+\lambda'\Hodge p}
\end{equation}
and $p^L={p'}^L$ (or equivalently $p^R={p'}^R$).
\end{lm}

\begin{proof}
Left covectors are preserved by Casimirs, thus we only need to find the vector field on the group. Let us denote the projection on the group manifolds of the Poisson vector fields of the Casimirs by $V_1$ and $V_2$.

We identify bivectors with the left covectors on $\SL$ by the scalar product and then
\begin{equation}
V_1=2\Lv(p^L),
\end{equation}
and thus $g$ is changed from the right (because left invariant vector field) by $p^L$.

The second Casimir is related to the first by Hodge star, thus
\begin{equation}
V_2=2\Lv(\Hodge p^L).
\end{equation}
Together (they commute) we have the flow 
\begin{equation}
g'=ge^{\lambda p^L+\lambda'\Hodge p^L}.
\end{equation}
From preservation of left covectors we have $p^L={p'}^L$.
\end{proof}

The symplectic reduction with respect to the Casimirs is given by
\begin{equation}
([g],B)\colon B\in X_{n,\rho},\quad [g]=[g'] \text{ if }\exists_{\lambda,\lambda'\in \R}
g=g'e^{\lambda B+\lambda'*B}.
\end{equation}
Let us denote
\begin{equation}
C_{n,\rho}=\{(g,B)\colon B\in X_{n,\rho}\}\subset T^*\SL.
\end{equation}
We have a map $\pi_{C_{n,\rho}}\colon C_{n,\rho}\rightarrow S$ to the symplectic reduction.

If $\LL'\subset S$ is a real lagrangean then
\begin{equation}
\pi^{-1}_{C_{n,\rho}}(\LL') 
\end{equation}
is also a lagrangean and it is a subset of $C_{n,\rho}$. The other way around, 
if a real lagrangean $\LL\subset T^*\SL$ is such that $\LL\subset C_{n,\rho}$, then as Casimir generated directions belong to $\LL$ we have
\begin{equation}
\LL=\pi^{-1}_{C_{n,\rho}}(\LL'),
\end{equation}
where $\LL'$ is a lagrangean in $S$.

The same holds for complex lagrangeans (in locally holomorphic extensions).

\subsubsection{Explicit description}

There is a direct description of this symplectic reduction that is an analog of Peter-Weyl theorem in group representation theory. Let us notice that the left and right invariant covectors Poisson commute with the Casimirs. Moreover the equation
\begin{equation}
p^R=-g^{-1}\cdot p^L
\end{equation}
has a solution for $g$ if $p^R$ and $p^L$ are of the same type (nonzero) and $g$ is unique up to $[\cdot]$ equivalence. Thus the map
\begin{equation}
([g],p)\rightarrow (p^L,p^R)\in X_{n,\rho}\times X_{n,\rho}
\end{equation}
is an isomorphism of symplectic spaces.\footnote{We used the fact that if $B\in X_{n,\rho}$ then also $-B\in X_{n,\rho}.$}

\subsection{Symplectic theory of $S_{ij}'$}

Let us consider an action
\begin{equation}
\tilde{S}_{ij}(g_i,g_j,\z_{ij},\z_{ji})=S_{ij}^{n_{ij}}(g_i^{-1}\z_{ij}) +S^\beta_{ij}(\z_{ij},\z_{ji})+S_{ji}^{n_{ji}}(g_j^{-1}\z_{ji}).
\end{equation}
Let us now assume that for every $ij$ the lagrangean
\begin{equation}
\LL_{ij}'=\LL_{S_{ij}^{n_{ij}}}'\subset X_{2\jj_{ij},\rho_{ij}}
\end{equation}
is strictly positive at the point corresponding to the fundamental stationary point (that is $[(g_i^0)^{-1}\z_{ij}^0]$). We will prove this fact in section \ref{sec:simplicity}.

Because the action $S^\beta$ is real, the imaginary part of the hessian with respect to $\z_{ij}$ and $\z_{ij}$ is block diagonal with respect to every $\z$ variable. From strict positivity of the lagrangean every block is strictly positive, thus by lemma \ref{lm:pos}, the form $H_{\z\z}$ is nondegenerate. 

We can now consider
\begin{equation}
S^{red}_{ij}(g_i,g_j)=S'_{ij}(g_{ij}).
\end{equation}
It is well defined for $g_{ij}$ in the neighbourhood of $g_{ij}^0$.

\begin{lm}
The lagrangean manifold of the action $S_{ij}'$ is given by
\begin{equation}
\LL_{S_{ij}'}=\pi_{C_{2\jj_{ij},\rho_{ij}}}^{-1}(\LL_{ij}'\times \LL_{ji}').
\end{equation}
\end{lm}

\begin{proof}
Left and right invariant derivatives of $S_{ij}'$  are equal to derivatives of $S_{ij}^{n_{ij}}$ and, respectively, $S_{ji}^{n_{ji}}$ with spinors equal to the stationary point solutions $\z^\C$
\begin{equation}
\delta^L S_{ij}'=\delta_{g_i}^LS_{ij}^{red}=\delta_\z S_{ij}^{n_{ij}}(g_i^{-1}\z_{ij}^\C),\quad 
\delta^R S_{ij}'=\delta_{g_j}^LS_{ij}^{red}=\delta_\z S_{ji}^{n_{ij}}(g_j^{-1}\z_{ji}^\C),
\end{equation}
because derivatives with respect to $\z$ vanish in the point $[\z^\C](g_i,g_j)$. We see from the type of the actions that $\LL_{S_{ij}'}\subset C_{2\jj_{ij}\rho_{ij}}^\C$, thus it is an inverse image of a complex lagrangean in $X_{2\jj_{ij},\rho_{ij}}\times X_{2\jj_{ij},\rho_{ij}}$. We see also that
\begin{equation}
\pi_{C_{2\jj_{ij}\rho_{ij}}}(\LL_{S_{ij}'})\subset \LL_{ij}'\times \LL_{ji}',
\end{equation}
and by comparing dimension it needs to be equal.
\end{proof}

Let us denote 
\begin{equation}
B_{ij}^{\zero}=\delta_\z S_{ij}^{n_{ij}}((g_i^0)^{-1}\z_{ij}^0).
\end{equation}
Let us notice $B_{ij}^{\zero}=(g_i^0)^{-1}B_{ij}^0$.

\begin{lm}
If every lagrangean $\LL_{ij}'$ is strictly positive then if for $v\in so(1,3)$ 
\begin{equation}
(\partial^2\Im S_{ij}')v=0,
\end{equation}
then $v\in\{B_{ij}^\zero,\Hodge B_{ij}^{\zero}\}$.
\end{lm}

\begin{proof}
From the previous lemma 
\begin{equation}
\LL_{S_{ij}'}=\pi_{C_{2\jj_{ij},\rho_{ij}}}^{-1}(\LL_{ij}'\times\LL_{ji}').
\end{equation}
Let $V=\Pi_{S_{ij}'}^{-1}(v)$ be the lift of $v$ to $T\LL_{S_{ij}'}$, its image 
\begin{equation}
\pi_{C_{2\jj_{ij},\rho_{ij}}}(V)\in T(\LL_{ij}'\times\LL_{ji}')\cap \overline{T(\LL_{ij}'\times\LL_{ji}')}=\{0\}.
\end{equation}
Thus $V$ is in the space of the Casimirs' Poisson vector fields. Thus its projection onto the tangent space of the group
\begin{equation}
v\in \{B_{ij}^{\zero},\Hodge B_{ij}^{\zero}\}
\end{equation}
as stated.
\end{proof}

\section{Simplicity constraints}\label{sec:simplicity}

Our goal in this section is to show that $\LL_{ij}'$ is strictly positive at the extremal point coming from the fundamental stationary point. In fact it is a simple computation of a two dimensional matrix. However it is useful to describe this lagrangean (in the neighbourhood of this point). Let us notice that from the reality condition of the action we know that the lagrangean is positive.

\subsection{Conditions on the action}

Let us suppose that we have a function of the form
\begin{equation}
 G^N(\z)=f(\z)e^{iNS(\z)},
\end{equation}
defined and analytic for $\z\in U$.

We have an action of the group on spinors $\z$, thus we can also consider an operator
\begin{equation}
\hat{D}=\sum_{|I|\leq m} (-i)^{|I|}d_{|I|}^{I_1\cdots I_{|I|}}\Lv_{\SP^+}(L_{I_1})\cdots \Lv_{\SP^+}(L_{I_{|I|}}),
\end{equation}
where $L_{I}$ are Lie algebra basis.

We associate with this operator a symbol (a homogenous polynomial on the Lie coalgebra)
\begin{equation}
P_D(p)=\sum_{|I|=m} d_{m}^{I_1\cdots I_m}L_{I_1}(p)\cdots L_{I_m}(p),
\end{equation}
where $p$ are Lie coalgebra elements. Let us remind that we identify both Lie algebra and coalgebra with bivectors (thanks to the scalar product).

Let  $p(\lambda)$ be a polynomial of order $m$ with $m$-homogeneous coefficient $a_m$ such that for every $N$
\begin{equation}
\left(\hat{D}-p(N)\right)G^N(\z)=0.
\end{equation}
Then taking the leading term in the $N$ expansion, we get for any $\z$
\begin{equation}
P_D\left(\delta_\z S\right)=a_m.
\end{equation}

\subsection{Bivector decomposition}\label{sec:biv-decomp}

For the given normal $N^0_i$ (see \cite{Kaminski2018}) with the norm $c_i=|N_i^0|^2\in\{-1,1\}$ we can decompose the bivector $B$ as follows
\begin{equation}
B=\Hodge (v\wedge N^0_i)+ w\wedge N^0_i,
\end{equation}
where $v,w\in {N^0_i}^\perp$ and the two terms belong to
\begin{equation}
so({N^0_i}^\perp)\oplus \Hodge so({N^0_i}^\perp).
\end{equation}
We can now introduce maps
\begin{align}
\Li_i\colon so(1,3)\rightarrow {N^0_i}^\perp,\quad \Li_i(B)=v,\\
\Ki_i\colon so(1,3)\rightarrow {N^0_i}^\perp,\quad \Ki_i(B)=w.
\end{align}
They are explicitly given by
\begin{equation}
\Li_i(B)=c_iN^0_i\llcorner\Hodge B,\quad  \Ki_i(B)=-c_iN^0_i\llcorner B.
\end{equation}
We can identify $so({N^0_i}^\perp)$ with the vector space ${N^0_i}^\perp$ by the map $\Li_i$
\begin{equation}
[\Hodge (v\wedge N_i^0),\Hodge (v'\wedge N_i^0)]=\Hodge \left((v\times v')\wedge N_i^0\right),
\end{equation}
where $\times$ is defined by
\begin{equation}
v\times v'=\Hodge (v\wedge N_i^0)(v')=-\Hodge \left(v\wedge v'\wedge N_i^0\right).
\end{equation}
The Casimirs can be writen in terms of these vectors as follows
\begin{equation}\label{eq:Casimirs-LK}
C_1=(B,B)=-c_i\left(\Li_i^2-\Ki_i^2\right),\quad C_2=(B,\Hodge B)=-2c_i\Li_i\cdot\Ki_i.
\end{equation}
With the vector $v\in {N^0_i}^\perp$ we can associate two complex vectors $k_s^i(v)$ ($s=\pm 1$) given by the conditions:
\begin{enumerate}
\item $k_s^i(v)\cdot N_0^i=k_s^i(v)\cdot v=k_s^i(v)\cdot k_s^i(v)=0$.
\item The action of the vector on $k_s$
\begin{equation}
v\times k_s^i(v)= i s Ck_s^i(v),
\end{equation}
where $C=\sqrt{(\Hodge (v\wedge N_i^0),\Hodge (v\wedge N_i^0))}=\sqrt{-c_iv\cdot v}$.
\end{enumerate}
In the case of spacelike faces we choose $C>0$. In this situation vectors $k_{\pm 1}^i(v)$ are complex and we assume
\begin{equation}
k_{- 1}^i(v)=\overline{k_{1}^i(v)}.
\end{equation}
With the choice of signature $(+---)$ the hermitian form
$\overline{w}\cdot w$
on $\{N_i^0,v\}^\perp$ is negatively definite thus, $k_{1}^i(v)\cdot k_{-1}^i(v)<0$. We assume that $k_{1}^i(v)\cdot k_{-1}^i(v)=-1$, and this fixes vectors up to a phase.

\begin{lm}
We have
\begin{equation}
k_{1}^i(v)\times k_{-1}^i(v)=i\frac{c_i}{C} v.
\end{equation}
\end{lm}

\begin{proof}
Let us notice that $k_{1}^i(v)\times k_{-1}^i(v)=\alpha v$ and
\begin{align}
&iCk_{-1}^i(v)\cdot k_{1}^i(v)=k_{-1}^i(v)\cdot (v\times k_{1}^i(v))=\\
&=-k_{-1}^i(v)\cdot ( k_{1}^i(v)\times v)=(k_{1}^i(v)\times k_{-1}^i(v))\cdot  v=
\alpha(v\cdot v).
\end{align}
Thus
\begin{equation}
k_{1}^i(v)\times k_{-1}^i(v)=i\frac{C(k_{-1}^i(v)\cdot k_{1}^i(v))}{v\cdot v} v,
\end{equation}
and substituting $v\cdot v=-c_iC^2$ we get the result.
\end{proof}

Let us notice that if a complex vector $w\in {N_i^0}^\perp$ satisfies $w\cdot w=w\cdot v=0$ then
\begin{equation}
w\in \Span\{k^i_1(v)\}\cup \Span\{k^i_{-1}(v)\}.
\end{equation}
We can regard $v\cdot \Li_i$ and $v\cdot \Ki_i$ as linear maps on bivectors, thus we can compute Poisson brackets. In order to do it we need to find the associated by (the scalar product) bivectors
\begin{equation}
v\cdot \Li_i(B)=(-c_i\Hodge (v\wedge N_i^0),B),\quad v\cdot \Ki_i(B)= (c_iv\wedge N_i^0,B),
\end{equation}
thus we get
\begin{align}
&\{v\cdot \Li_i,v'\cdot \Li_i\}=c_i(v\times v')\cdot \Li_i,\\
&\{v\cdot \Li_i,v'\cdot \Ki_i\}=c_i(v\times v')\cdot \Ki_i,\\
&\{v\cdot \Ki_i,v'\cdot \Ki_i\}=-c_i(v\times v')\cdot \Li_i.
\end{align}

\subsection{Simplicity constraints}\label{sec:simplicity}

The coherent states $\Phi^{n_{ij}}(\z_{ij})$ satisfies the following equations
\begin{enumerate}
\item \textit{Diagonal simplicity constraints}, that for fixed spins means that the values of the Casimir operators are related to twisted simplicity constraints\footnote{Quantisation of  the action of the Lie algebra element $L$ is $\hat{L}=\frac{1}{i}\Lv_{\SP^+}(L)$.}
\begin{equation}
\hat{C}_1=\frac{1}{4}(n^2-\rho^2-4),\ \hat{C}_2=-\frac{1}{2}n\rho,
\end{equation}
where $\rho=\gamma n$ and $n=2\jj_{ij}$.\footnote{Our convention differs  from \cite{Conrady2010} by a sign in $C_2$ that can be seen from \eqref{eq:Casimirs-LK}.}
\item \textit{Cross simplicity constraints}, that are implemented in the EPRL model by
\begin{align}
&\left(\gamma\hat{\Li}_i+\hat{\Ki}_i\right)^2=0,\\
&\left(\hat{\Li}_i-\gamma\hat{\Ki}_i\right)\cdot\left(\gamma\hat{\Li}_i+\hat{\Ki}_i\right)=0.
\end{align}
\item \textit{The coherent state condition} $k_{s_{ij}}(v_{ij})\cdot \hat{\Li}_i=0$, where $s_{ij}$ is fixed and $v_{ij}$ is constructed from $n_{ij}$.
\end{enumerate}
These conditions impose several conditions on $S_{ij}^{n_{ij}}$. We can describe them in terms of $\LL_{ij}'$. Namely $B\in \LL_{ij}'$ needs to satisfy
\begin{enumerate}
\item \textit{Diagonal simplicity constraints} $(B,B)=\frac{1}{4}(4\jj_{ij}^2-\rho_{ij}^2)$ and $(B,\Hodge B)=-\frac{1}{2}2\jj_{ij}\rho_{ij}$ that are satisfied because $B\in X_{2\jj_{ij},\rho_{ij}}$.
\item \textit{Cross simplicity constraints} 
\begin{align}
&\left(\gamma\Li_i+\Ki_i\right)^2=0,\\
&\left(\Li_i-\gamma\Ki_i\right)\cdot\left(\gamma\Li_i+\Ki_i\right)=0.
\end{align}
\item \textit{Coherent state condition} $k_{s_{ij}}(v_{ij})\cdot \Li_i=0$, where $s_{ij}$ is fixed and $v_{ij}$ is constructed from $n_{ij}$.
\end{enumerate}
In order to analyze the conatraints let us  introduce a twisting map
\begin{equation}
\tau\colon so(1,3)\rightarrow so(1,3),\quad \tau(B)=B+\gamma\Hodge B.
\end{equation}
We can compute
\begin{align}
(\tau(B),\tau(B))&=(1-\gamma^2)(B,B)+2\gamma (B,\Hodge B),\\
(\tau(B),\Hodge \tau(B))&=(1-\gamma^2)(B,\Hodge B)-2\gamma (B,B).
\end{align}
Similarly
\begin{equation}
\Li_i(\tau(B))=\Li_i(B)+\gamma\Ki_i(B),\quad \Ki_i(\tau(B))=\Ki_i(B)-\gamma\Li_i(B).
\end{equation}
Let us denote $B^\tau=\tau^{-1}(B)$ and $\Li_i^\tau(B)=\Li_i(B^\tau)$, $\Ki_i^\tau(B)=\Ki_i(B^\tau)$, then
\begin{equation}
\Li_i(B)=\Li_i^\tau(B)+\gamma\Ki_i^\tau(B),\quad \Ki_i(B)=\Ki_i^\tau(B)-\gamma\Li_i^\tau(B).
\end{equation}
The first two conditions mean
\begin{enumerate}
\item Diagonal simplicity conditions:
\begin{equation}
(B^\tau,B^\tau)=\jj_{ij}^2,\quad (B^\tau,\Hodge B^\tau)=0.
\end{equation}
\item Cross simplicity: $\Ki_i^\tau\in \Span\{k_{1}^i(\Li_i^\tau)\}\cup \Span\{k_{-1}^i(\Li_i^\tau)\}$.
\end{enumerate}
Thus we can write
\begin{equation}
B^\tau=\Hodge (v\wedge N_i^0)+\lambda' k_t^i(v)\wedge N_i^0
\end{equation}
and the Casimir conditions means that
\begin{equation}
-c_i|v|^2=\jj_{ij}^2.
\end{equation}
We are interested in the fundamental stationary point, and then $(B^\zero_{ij})^\tau=\Hodge v_{ij}\wedge N_i^0$. The space $\LL_{ij}'$ around this point is a manifold thus there is a choice $t_{ij}$ such that 
\begin{equation}
\Ki_i^\tau\in \Span\{k_{t_{ij}}^i(\Li_i^\tau)\}.
\end{equation}
We also have 
\begin{equation}
-i\{\overline{k_{t_{ij}}^i(v)}\cdot(\Ki_i+\gamma\Li_i), k_{t_{ij}}^i(v)\cdot(\Ki_i+\gamma\Li_i)\}=-t_{ij}\frac{v}{C}\cdot ((\gamma^2-1)\Li_i+2\gamma\Ki_i),
\end{equation}
and, from positivity of the lagrangean, the right hand side needs to be positive. Let us notice that
\begin{equation}
\frac{v}{C}\cdot ((\gamma^2-1)\Li_i+2\gamma\Ki_i)=
(1+\gamma^2)\frac{v}{C}\cdot (\gamma\Ki_i^\tau-\Li_i^\tau).
\end{equation}
As at $B=B^\zero_{ij}$ we have ($C=\jj_{ij}$)
\begin{equation}
v_{ij}\cdot \Li_i^\tau=|v_{ij}|^2=-c_iC^2,\quad  v_{ij}\cdot \Ki_i^\tau=0,
\end{equation}
we see that $t_{ij}=-c_i$.

Let us consider now coherent state condition $k_{s_{ij}}(v_{ij})\cdot \Li_i=0$. It means that
\begin{equation}
\Li_i=\lambda_1 v_{ij}+\lambda_2 k_{s_{ij}}(v_{ij}).
\end{equation}
However,
\begin{equation}
\Li_i^2=(\lambda_1 v_{ij}+\lambda_2 k_{s_{ij}}(v_{ij}))^2=\lambda_1^2|v_{ij}|^2.
\end{equation}
but $\Li_i^\tau\cdot\Ki_i^\tau=\Ki_i^\tau\cdot \Ki_i^\tau=0$, thus
\begin{equation}
\Li_i^2=(\Li_i^\tau+\gamma \Ki_i^\tau)^2=(\Li_i^\tau)^2=-c_i [-c_i((\Li_i^\tau)^2-(\Ki_i^\tau)^2)]=-c_i\jj_{ij}^2,
\end{equation}
and $\lambda_1 =\pm 1$. As the phase space point corresponding to the fundamental stationary point is in the lagrangean we have in the neighbourhood of this stationary point $\lambda_1=1$. We can now compute
\begin{equation}
-i\{k_{-s_{ij}}(v_{ij})\cdot \Li_i,k_{s_{ij}}(v_{ij})\cdot \Li_i\}=-s_{ij}\frac{v_{ij}}{C}\cdot \Li_i=-s_{ij}\frac{v_{ij}}{C}\cdot (\Li_i^\tau+\gamma \Ki_i^\tau),
\end{equation}
at the fundamental stationary point it is equal to $s_{ij}c_i\jj_{ij}$, thus $s_{ij}=c_i$. 

We can now describe tangent space to the lagrangean at $B_{ij}^\zero$. The conditions for bivectors to be tangent directions to $\LL_{ij}'$ is that
\begin{enumerate}
\item Tangency condition $(B, B_{ij}^\zero)=0$, $(B,\Hodge B_{ij}^\zero)=0$ (this is equivalent to $v_{ij}\cdot \Li_i(B)=v_{ij}\cdot \Ki_i(B)=0$ and also $v_{ij}\cdot \Li_i^\tau(B)=v_{ij}\cdot \Ki_i^\tau(B)=0$),
\item $k_{t_{ij}}(v_{ij})\cdot (\Ki_i(B)+\gamma \Li_i(B))=0$ (that is $k_{t_{ij}}(v_{ij})\cdot \Ki_i^\tau(B)=0$),
\item $k_{s_{ij}}(v_{ij})\cdot \Li_i(B)=0$.
\end{enumerate}
It is not hard to find all vectors satisfying these conditions. Every tangent bivector can be  uniquely described by a pair
\begin{equation}
B\rightarrow (\Li_i(B),\Ki_i^\tau(B)).
\end{equation}
The conditions on $B$ are
\begin{equation}
(\Li_i(B),\Ki_i^\tau(B))=(\lambda_s k_{s_{ij}}(v_{ij}),\lambda_tk_{t_{ij}}(v_{ij})),\quad \lambda_s,\lambda_t\in \C.
\end{equation}
We can now summarize

\begin{lm}
At the point $B_{ij}^\zero$ the lagrangean $\LL_{ij}'$ is strictly positive.
\end{lm}

\begin{proof}
We need to prove that the real tangent vector (bivector) is zero. Tangent vectors satisfy
\begin{equation}
(\Li_i(B),\Ki_i^\tau(B))=(\lambda_s k_{s_{ij}}(v_{ij}),\lambda_tk_{t_{ij}}(v_{ij})),\quad \lambda_s,\lambda_t\in \C,
\end{equation}
and from reality
\begin{align}
&(\Li_i(B),\Ki_i^\tau(B))=(\Li_i(\bar{B}),\Ki_i^\tau(\bar{B}))=\\
&=(\overline{\lambda_s k_{s_{ij}}(v_{ij})},\overline{\lambda_tk_{t_{ij}}(v_{ij})})=(\overline{\lambda_s} k_{-s_{ij}}(v_{ij}),\overline{\lambda_t}k_{-t_{ij}}(v_{ij})).
\end{align}
However vectors $k_{\pm 1}^i(v_{ij})$ are linearly independent thus $\lambda_s=\lambda_t=0$.
\end{proof}

\section{Reduced hessian}\label{sec:reduced-hes}

Let us denote the tensor of second derivatives of $\Im S_{ij}(g_{ij}^{0})$ by $I_{ij}'$. We are interested in the second derivatives $\Im S^{red}_{ij}(g_i,g_j)$ at $\{g_k^{0}\}$ (we assume $g_5^{0}=1$).

The tangent vectors to the manifold $\prod_{i=1}^4 \SL$ are given by
\begin{equation}
V_{tot}=\{v\colon \{1,\ldots 5\}\rightarrow so(1,3)\colon v(5)=0\}.
\end{equation}
For convenience we assumed $v(5)=0$. We use here the right invariant vector fields to identify $V_{tot}$ with $T\left(\prod_{i=1}^4 \SL\right)$.

\begin{lm}
We have
\begin{equation}
I_{ij}(v,v'):=\partial^2\Im S_{ij}^{red}(\{g_k^{0}\})(v,v')= {I}'_{ij}(g_i^{-1}\cdot v(i)-g_i^{-1}\cdot v(j),g_i^{-1}\cdot v(i)-g_i^{-1}\cdot v(j)),
\end{equation}
where $I'_{ij}(L,L)=\partial^2 \Im S'_{ij}(\Lv(L),\Lv(L))$.
\end{lm}

\begin{proof}
Standard result about functions of the form $f(g_j^{-1}g_i)$. We use the right invariant vector fields, thus
\begin{equation}
g_i=e^{-tL_i}g_i ,\quad g_j=e^{-tL_j}g_j. 
\end{equation}
We have for left invariant vector fields
\begin{equation}
\partial^2 f(L,L)=\frac{d^2}{dt^2}|_{t=0} f(ge^{tL}).
\end{equation}
We can now compute for $v$
\begin{equation}
v(i)=L_i,\ v(j)=L_j,\ v(k)=0\text{ for } k\not=i,j.
\end{equation}
Let us compute second derivative of $F(g_i,g_j)=f(g_{ij})$ where $g_{ij}=g_j^{-1}g_i$
\begin{equation}
\partial^2 F(v,v)=\frac{d^2}{dt^2}|_{t=0} f(g_j^{-1}e^{tL_j}e^{-tL_i}g_i)=\frac{d^2}{dt^2}|_{t=0} f(g_j^{-1}g_ie^{X(t)}).
\end{equation}
We used BCH formula and commuted (we use notation $g\cdot L=gLg^{-1}$)
\begin{align}
&e^{X(t)}=e^{tg_i^{-1}L_jg_i} e^{-tg_i^{-1}L_ig_i},\\
&X(t)=tg_i^{-1}\cdot L_j-tg_i^{-1}\cdot L_i-\frac{t^2}{2}[g_i^{-1}\cdot L_j,g_i^{-1}\cdot L_i]+O(t^3).
\end{align}
We use now $\partial \Im S'_{ij}=0$ to get
\begin{equation}
\partial^2 \Im S_{ij}^{red}(v,v)=\partial^2\Im S'_{ij}(g_i^{-1}\cdot v(i)-g_i^{-1}\cdot v(j),g_i^{-1}\cdot v(i)-g_i^{-1}\cdot v(j)),
\end{equation}
so we found the desired result.
\end{proof}

\begin{lm}\label{lm-no}
Let us suppose that $0\not=v\in V_{tot}$ satisfies
\begin{equation}
\forall_{1\leq i<j\leq 5} I_{ij}(v,v)=0,
\end{equation}
then there exist $1\leq a<b\leq 4$ such that the bivectors
\begin{equation}
B_{a5}^{0},\ \Hodge B_{a5}^{0},\ B_{b5}^{0},\ \Hodge B_{b5}^{0},\ B_{ab}^{0},\ \Hodge B_{ab}^{0},
\end{equation}
are linearly dependent.
\end{lm}

\begin{proof}
As all $I_{ij}$ are positive definite and $I_{ij}'$ has the kernel spanned by
\begin{equation}
B_{ij}^{\zero},\ \Hodge B_{ij}^{\zero},
\end{equation}
we have
\begin{equation}
I_{ij}(v,v)=0\Leftrightarrow v(i)-v(j)\in \Span\{B_{ij}^{0},\ \Hodge B_{ij}^{0}\}.
\end{equation}
We see that from $I_{i5}(v,v)=0$ it follows that $v(i)\in \Span\{B_{i5}^{0},\ \Hodge B_{i5}^{0}\}$ and thus as $v$ is nonzero there exist $i,j\not=5$ such that
\begin{equation}
0\not=v(i)\in \Span\{B_{i5}^{0},\ \Hodge B_{i5}^{0}\}
\end{equation}
and also
\begin{equation}
v(i)=(v(i)-v(j))+v(j)\in \Span\{B_{ij}^{0},\ \Hodge B_{ij}^{0},B_{j5}^{0},\ \Hodge B_{j5}^{0}\}.
\end{equation}
This means linear dependence.
\end{proof}

\begin{lm}\label{lm-rec}
If the reconstructed $4$-simplex (in any signature) with spacelike faces is nondegenerate then
\begin{equation}
B_{a5}^{\geom},\ \Hodge B_{a5}^{\geom},\ B_{b5}^{\geom},\ \Hodge B_{b5}^{\geom},\ B_{ab}^{\geom},\ \Hodge B_{ab}^{\geom},
\end{equation}
are linearly independent for $\{a,b,5\}$ distinct.
\end{lm}

\noindent Here by $B_{ij}^\geom$ we denote geometric bivectors of the reconstructed $4$-simplex (see \cite{Kaminski2018}).

\begin{proof}
Let us assume $a=3$, $b=4$. The bivectors $B_{ij}^{\geom}$ for $i,j\in \{3,4,5\}$ can be written as
\begin{equation}
B_{ij}^{\geom}=\eta_{ij}\wedge e_{12},
\end{equation}
where $\eta_{ij}\perp e_{12}$ and $e_{12}$ is the edge vector connecting vertex $1$ with $2$ (this edge is spacelike). Moreover $\eta_{ij}$ are independent if the $4$-simplex is nondegenerate.

Let us notice that $e_{12}\llcorner\Hodge B_{ij}^{\geom}=0$ and $e_{12}\llcorner B_{ij}^{\geom}=-\eta_{ij}|e_{12}|^2$.

Let us assume that there is a linear equation for the bivectors 
\begin{equation}
\sum_{ij\in\{3,4,5\}} \lambda_{ij}B_{ij}^{\geom}+\lambda_{ij}'\Hodge B_{ij}^{\geom}=0.
\end{equation}
Contracting it with $e_{12}$ we get
\begin{equation}
\sum_{ij\in\{3,4,5\}} \lambda_{ij}\eta_{ij}=0\Rightarrow \lambda_{ij}=0.
\end{equation}
Taking the Hodge dual of the equation and then contracting with $e_{12}$ we get
\begin{equation}
\sum_{ij\in\{3,4,5\}} \lambda_{ij}'\eta_{ij}=0\Rightarrow \lambda_{ij}'=0.
\end{equation}
Thus the bivectors are linearly independent if the reconstructed $4$-simplex is nondegenerate.
\end{proof}

\begin{thm}
The reduced hessian for the lorentzian EPRL model (and for the Conrady-Hnybida extension for spacelike faces) is nondegenerate at the stationary point that corresponds to a nondegenerate $4$-simplex.
\end{thm}

\begin{proof}
If $Hv=0$ then we are in the situation from lemma \ref{lm-no}. From positivity of the $I_{ij}$ it thus follows that
\begin{equation}\label{eq:linear-dep}
B_{a5}^{0},\ \Hodge B_{a5}^{0},\ B_{b5}^{0},\ \Hodge B_{b5}^{0},\ B_{ab}^{0},\ \Hodge B_{ab}^{0},
\end{equation}
are linearly dependent.

Let us consider now separetely two cases:
\begin{enumerate}
\item If the stationary point corresponds to a lorentzian $4$-simplex then
\begin{equation}
B_{ij}^{0}=\tau(B_{ij}^\geom)
\end{equation}
and $\tau$ preserves the space \eqref{eq:linear-dep}. By lemma \ref{lm-rec} we have a contradiction.
\item If the stationary point ($+$) corresponds to a $4$-simplex solution with other signature then there is the second point ($-$) and
\begin{equation}
B_{ij}^+=\tau(\Hodge v_{ij}^+\wedge N^0),\quad B_{ij}^-=\tau(\Hodge v_{ij}^-\wedge N^0),
\end{equation}
and $B_{ij}^\geom$ has selfdual and antiselfdual parts given by $\tau^{-1}(B_{ij}^\pm)$. From \eqref{eq:linear-dep} it follows that there exist constants $\lambda_{ij},\lambda_{ij}'$ such that
\begin{equation}
\sum_{ij\in\{a,b,5\},\ i<j} \lambda_{ij}B_{ij}^++\lambda_{ij}'\Hodge B_{ij}^+=0,
\end{equation}
thus taking $\Li_i^\tau$ and $\Ki_i^\tau$ parts we get
\begin{equation}
\sum_{ij\in\{a,b,5\},\ i<j} \lambda_{ij}v_{ij}^+=0,\quad 
\sum_{ij\in\{a,b,5\},\ i<j} \lambda_{ij}'v_{ij}^+=0.
\end{equation}
As some coefficients need to be nontrivial we get that $v_{ij}^+$ and thus also
$B_{ij}^+$ ($(i,j)\in \{a,b,5\}$) are linearly dependent. But this means that
\begin{equation}
B_{a5}^{\geom},\ \Hodge B_{a5}^{\geom},\ B_{b5}^{\geom},\ \Hodge B_{b5}^{\geom},\ B_{ab}^{\geom},\ \Hodge B_{ab}^{\geom},
\end{equation}
are linearly dependent and from lemma \ref{lm-rec} we have a contradiction.
\end{enumerate}
Independently of the signature of the reconstructed $4$-simplex the hessian is nondegenerate.
\end{proof}

\section{Summary}

We showed that the hessian in the EPRL and Conrady-Hnybida (spacelike surfaces case) is nondegenerate for any stationary point (corresponding to a nondegenerate $4$-simplex of either lorentzian, euclidean or split singature). We also showed nondegeneracy  for the euclidean $\gamma<1$ case.  Our method works fine also for $\gamma>1$, but we have not provided the details in this case. However, the method does not extend immediately to the situation when some of the faces are timelike (the asymptotic of this case was considered recently in \cite{Liu2019}). The action in this case is purely real, and as we based our proof on the properties of imaginary part of the action, this case cannot be covered with the tools used in our paper unless they will be properly modified. The issue deserves a separate treatment and we leave this topic for future research.

\medskip
\noindent{\textbf{Acknowledgements:}}
We thank Marcin Kisielowski for fruitful discussions at the early stage of this project.

\appendix

\section{Notation}
\label{sec:appendix}

In this section we collect our notation:
\begin{enumerate}
\item The signature of the metric is $(+---)$.
\item Bivectors $so(1,3)=\Lambda^2 \R^4$ (we use identification by the scalar product). The action on vectors can be expressed as
\begin{equation}
(v\wedge v') (w)=v(v'\cdot w)-v' (v\cdot w).
\end{equation}
 We define also a scalar product $(\cdot,\cdot)$ on bivectors
\begin{equation}
(v\wedge w, v'\wedge w')=\det\left(\begin{array}{cc}
    v\cdot v' &v\cdot w'\\ w\cdot v' & w\cdot w'\end{array}\right).
\end{equation}
Hodge star operation is denoted by $\Hodge$.
\item The adjoint action on the Lie algebra is defined by 
\begin{equation}
g\cdot L=gLg^{-1}.
\end{equation}
Coadjoint action on $P$ is defined by $g\cdot P (L)=P(g^{-1}\cdot L)$.
\item The hessian is a symmetric two form (tensor) on the tangent vectors for a function $f$ at the point where the first derivative vanishes. We denote this form by $\partial^2 f$.
\item The stationary point $\{g_i^0,[\z_{ij}^0]\}$, $g_5^0=1$ of the total action is referred to as the fundamental stationary point. The bivectors at this stationary point are denoted $B_{ij}^0$ (in the simplex frame) and $B_{ij}^\zero=(g_i^0)^{-1}B_{ij}^0$ (see the beginning of section \ref{sec:extension}). The geometric bivectors $B_{ij}^\geom$ are described in \cite{Kaminski2018} and appear in section \ref{sec:reduced-hes}.
\item Weyl spinor spaces ${\SP}^\pm=\C^2$: We denote spinors from $\SP^+$ by $\z$, $\v$, $\u$ etc. Clifford elements for any vector $v$ are $\eta^\pm(v)\colon \SP^\pm\rightarrow {\mathcal S}^\mp$ fulfilling
\begin{equation}
\eta^\mp(v)\eta^\pm(v')+\eta^\mp(v')\eta^\pm(v)=(v\cdot v'){\mathbb I}_{\SP^\pm}.
\end{equation}
The Lie algebra isomorphism
$so(1,3)$ to $sl(2,\C)$ (traceless matrices) is 
\begin{equation}
B\rightarrow \M(B), \quad \M(v\wedge v')=\frac{1}{4}\left(\eta^-(v)\eta^+(v')-\eta^-(v')\eta^+(v)\right).
\end{equation}
For two spinors $\u,\v$ we denote
\begin{equation}
[\u,\v]=\u^T\omega \v,\quad \omega=\bmat{0 &1\\ -1 &0}.
\end{equation}
Every traceless matrix can be written as ${\mathbb N}=\frac{1}{2}(\u_+\u_-^T+\u_-\u_+^T)\omega$ and
\begin{equation}
{\mathbb N}(\u_\pm)=\pm\frac{1}{2}[\u_-,\u_+]\u_\pm
\end{equation}
(see section \ref{sec:coadjoint} and \cite{Kaminski2018}).
\item $p^L$ and $p^R$ are left and right covectors (see section \ref{symp:G}) and 
\begin{equation}
\delta^LS(L)=\Lv(L)S,\quad \delta^RS(L)=\Rv(L)S,\quad 
\end{equation}
where $\Lv$ and $\Rv$ are left and right derivatives. We also denote $\delta=\delta^L$.
\item We denote
\begin{equation}
\delta_\z f(\z)=\delta^L f(g^{-1}\z)|_{g=1}
\end{equation}
(see section \ref{symp:G}).
\item $[\cdot]$ is a relation on spinors
\begin{equation}
[\z]=[\w]\Leftrightarrow \exists 0\not=\lambda \in \C\colon \z=\lambda \w,
\end{equation}
thus $[\z]$ is a point of $\CP$, $[\z_{ij}^\C]$ is a point on complexified $\CP^\C$.
\item The vector fields of the action of $\SL$ on $\CP$ are denoted by $\Lv_\CP(L)$ for $L\in so(1,3)$. They correspond to the curves
\begin{equation}
t\rightarrow [e^{-tL}\z].
\end{equation}
Similarly, the vector field of the action of $\SL$ on $\SP^+$ are denoted by $\Lv_{\SP^+}(L)$ for $L\in so(1,3)$. They correspond to the curves
\begin{equation}
t\rightarrow e^{-tL}\z.
\end{equation}
\item The definition of $\Li_i$ $\Ki_i$  is in section \ref{sec:biv-decomp}.
For the twisting map $\tau$, and twisted versions $\Li^\tau$, $\Ki_i^\tau$ see section \ref{sec:simplicity}.
\item The vectors $k_{\pm 1}^i(v)$ are defined in \ref{sec:biv-decomp}.
\item $X_{n,\rho}$ is a coadjoint orbit space defined in equation \eqref{eq:X}.
\item The projection from $X_{n,\rho}$ (coadjoint orbit) to $\CP$ is denoted by $\pi$. The function $f$ that is constant along the fibers can be pushed forward to $\CP$ and  such push forward is denoted by $[f]_\CP$ (see section \ref{sec:coadjoint}).
\item $S^{red}$, $S_{ij}^{red}$ are defined in section \ref{sec:reduced}. Their hessians are denoted by $H^{red}$ and $H_{ij}^{red}$.
\item $S'_{ij}$ and $g_{ij}=g_j^{-1}g_i$ is defined in section \ref{sec:reduced}.
\item $\LL$ denotes lagrangeans. The subscript denotes the (part of the) action generating the given lagrangean. We use $'$ to indicate lagrangeans in the coadjoint orbit space.
\item The form $\RR$ on the tangent space of the lagrangean at the real point is defined in equation \eqref{eq:R}.
\end{enumerate}

%%%%%%%%%%%%%%%%%%%%%%%%%%%%%%%%%%%%5

\bibliography{bibliography}{}
\bibliographystyle{ieeetr}

\end{document}